\documentclass[11pt]{amsart}

\usepackage{amssymb,mathrsfs}
\usepackage{graphicx}
\usepackage{algorithmic}
\usepackage{subfigure}

\usepackage{enumerate}        

\graphicspath{{Figures/}}

\newtheorem{theorem}{Theorem}[section]
\newtheorem{lemma}[theorem]{Lemma}

\newtheorem{corollary}[theorem]{Corollary}

\theoremstyle{definition}
\newtheorem{definition}[theorem]{Definition}

\newtheorem{example}[theorem]{Example}

\newcommand{\R}{\mathbb{R}}
\newcommand{\N}{\mathbb{N}}

\begin{document}

\title{Multidimensional online robot motion}
\author[J. Brown Kramer]{Josh~Brown~Kramer}
      \address{\tt Department of Mathematics and Computer Science\\
               Illinois Wesleyan University\\
               Bloomington IL  61701}
      \email{jbrownkr@iwu.edu}
\author[L. Sabalka]{Lucas~Sabalka}
      \address{\tt Department of Mathematical Sciences\\
               Binghamton University\\
               Binghamton NY  13902-6000
\newline       http://www.math.binghamton.edu/sabalka}
      \email{sabalka@math.binghamton.edu}

\begin{abstract}

We consider three related problems of robot movement in arbitrary dimensions: coverage, search, and navigation.  For each problem, a spherical robot is asked to accomplish a motion-related task in an unknown environment whose geometry is learned by the robot during navigation. The robot is assumed to have tactile and global positioning sensors.  We view these problems from the perspective of (non-linear) competitiveness as defined by Gabriely and Rimon.  We first show that in 3 dimensions and higher, there is no upper bound on competitiveness: every online algorithm can do arbitrarily badly compared to the optimal.  We then modify the problems by assuming a fixed clearance parameter.  We are able to give optimally competitive algorithms under this assumption.

\end{abstract}

\maketitle

\section{Introduction}\label{sec:intro}

This paper is about online sensor-based motion problems for robots in an unknown bounded $n$-dimensional environment.  Consider Bob, a spherical mobile robot with radius $r > 0$ at starting point $S$ in a space $X \subseteq \R^n$, where $X$ has finite diameter.  Bob is equipped with:

    \begin{itemize}
    \item a tactile sensor for feeling and tracing obstacle boundaries, and
    \item a precise global positioning sensor, which tells Bob its location using global coordinates on $X$.
    \end{itemize}
For our tasks, Bob will also be able to remember an amount of information proportional to the size of the space, but a priori Bob has no other knowledge of its surroundings.

For any point $p \in X$ and any fixed position for Bob, Bob is \emph{at} $p$ if $p$ is the location of Bob's center, and Bob \emph{occupies} $p$ if $p$ is within distance $r$ of Bob's center.

The three tasks we consider within this setup are:
    \begin{itemize}
    \item \emph{Cover}:  Describe an efficient way for Bob to move within $X$ to occupy every point in $X$ that can be occupied, and return to the starting point.  We denote this task by $COVER$, or $COVER_n$ if $n$ is known.
    \item \emph{Search}:   Given a target point $T$ with unknown coordinates (which is recognizable on contact), describe an efficient way for Bob to move within $X$ from $S$ to $T$.  We denote this task by $SEARCH$, or $SEARCH_n$ if $n$ is known.
    \item \emph{Navigate}:   Given a target point $T$ with known coordinates, describe an efficient way for Bob to move within $X$ from $S$ to $T$.  We denote this task by $NAV$, or $NAV_n$ if $n$ is known.
    \end{itemize}
To refer to one of these three tasks without specifying which, we will write $TASK$, or $TASK_n$ if $n$ is explicit.

The purpose of this paper is first to show that if $n > 2$ then, in a precise sense, there is no efficient algorithm to solve any of these problems.  We then show that with a minor modification to the problems, these tasks can be accomplished in an efficient manner, and we give efficient solutions.

Online motion algorithms in general are discussed frequently in robotics and computational geometry, and have been a recent active area of research.  There are many possible references to algorithms in this area, to which we only name the most relevant to our purposes.  For more detailed overviews, see for instance one of \cite{Berman, FiatWoeginger, Latombe, LaValle}.

As might be expected, sensor-based motion planning arises in a number of applications.  Examples include area coverage problems like cleaning public places, navigation problems like mail delivery in a city or moving packages in a factory, sample acquisition, and planetary exploration; the Mars rover uses autonomous online navigation algorithms \cite{ShenNagy}.

Results concerning online motion algorithms are almost always discussed in terms of the sensors with which the robots are equipped.  Often, but not always, robots are given visual sensors to be able to detect (nearby) objects within a line of sight.  However, problems requiring only tactile sensors do occur in situations where vision-based sensors are unrealistic.  For instance, navigation is often desired in abstract spaces, like the configuration space of a mechanical arm linkage, in which visual sensors, at least in their most literal interpretation, do not make sense.

The motion problems listed above have been frequently studied, but almost invariably in special instances.  Some of the earliest work on efficient robot motion is that of Lumelsky and Stepanov \cite{LumelskyStepanov}.  That work resulted in the BUG algorithms, which solve the $NAV_2$ problem for a (point) robot in the presence of arbitrary obstacles.  The BUG1 algorithm, described here in Section \ref{sec:CBUG}, was proven to run in time proportional to the lengths of perimeters of obstacles in $X$.  However, in terms of the length of the optimal path, BUG1 is not at all `competitive':  the path BUG1 takes can be arbitrarily long compared to an optimal path.  Thus, BUG1 is not competitive in the classical sense.  Papadimitriou and Yannakakis \cite{PapadimitriouYannakakis} provided the first competitive analysis of the $NAV_2$ problem in specific instances.  More recently, Gabriely and Rimon \cite{GabrielyRimon} have given a modification of BUG1, called CBUG, which is `optimally' competitive.  We also describe the CBUG algorithm in Section \ref{sec:CBUG}.  Gabriely and Rimon generalized the definition of competitiveness to characterize in what sense CBUG is optimally competitive -- namely, CBUG is quadratically competitive.  To analyze our algorithms, we use the Gabriely-Rimon definition of competitiveness as presented here in Section \ref{sec:competitiveness}.  Roughly, optimal competitiveness of an algorithm means that the path it generates has length, in the worst case, proportional to the optimal worst case length generated among all online navigators over all possible environments.  We quantify performance by measuring relative to the best offline path (ie the path generated by a robot with complete knowledge of $X$).  As with most algorithms for motion planning, both the BUG1 and CBUG algorithms are for $2$-dimensional spaces.

For the $SEARCH_2$ problem, a notable linearly competitive solution in a number of environments was given by \cite{BaezaYatesCulbersonRawlins}.

To the authors' knowledge, there are very few papers dealing with higher-dimensional sensor-based motion algorithms.  There is the paper of Cox and Yap \cite{CoxYap}, which extends the BUG algorithms to a 3-dimensional rod, and there are those papers using Choset and Burdick's Hierarchical Generalized Voronoi Graphs (HGVGs) \cite{ChosetBurdick} which work in higher dimensions but require visual sensors.  Roughly, an HGVG is a way of creating a roadmap along the lines of \cite{Canny} in higher dimensions.  These roadmaps are essentially a $1$-dimensional subspace of the navigable space (equidistant from $n-1$ obstacles), which can be created incrementally.  These HGVGs provide a nonheuristic navigation algorithm which is complete -- that is, is guaranteed to work.  However, Cox and Yap's results only apply to their particular problem, and Choset and Burdick assume visual sensors while not providing a bound on competitiveness.

It is clear that, for most environments $X$ of dimension $n \geq 3$, the $COVER_n$ task is actually impossible to solve.  For instance, assume that $X$ has a large codimension 1 cube as an obstacle.  Then Bob cannot possibly occupy every point very near the cube in finite time -- to do so, Bob's center would need to be at every point in a codimension-one cube (of distance $r$ from the original obstacle).  Thus, some slight modification of $COVER_n$ is necessary.

Interestingly, for $n \geq 3$ the optimal online distance for $SEARCH_n$ and $NAV_n$ can be arbitrarily bad compared to the optimal offline distance.  One result from this paper, to be made precise via Theorem \ref{thm:universallb} and Corollary \ref{cor:needmod}, is:
    \begin{theorem}
    If $n \geq 3$ then every algorithm that solves either $NAV_n$ or $SEARCH_n$ has no upper bound on competitiveness with respect to optimal length.
    \end{theorem}
Thus, some slight modifications of the $SEARCH_n$ and $NAV_n$ tasks are also necessary.

We modify $TASK_n$ to allow for a small amount of error, herein called the \emph{clearance parameter} $\epsilon$ (see Section \ref{sec:modifying}), which controls the narrowness of the paths we require Bob to follow.  We place no other constraints on our spaces $X$:  we do not require the obstacles be rectangles, polygons, convex, etc.  Although we modify $TASK_n$, our modifications can be physically negligible, as $\epsilon$ can be as small as desired.

We prove:
    \begin{theorem}[c.f. Theorem \ref{thm:universallb}]\label{thm:B}
    The modified $NAV_n$ and $SEARCH_n$ tasks have a universal lower bound on competitiveness with respect to optimal length $l_{opt}$ given by
        $$\frac{l_{opt}^n}{\kappa^{n-2}(r+\epsilon)},$$
where $\kappa = 2\sqrt{2r\epsilon + \epsilon^2}$.
    \end{theorem}
For small $\epsilon$, the value $\kappa$ is approximately $2\sqrt{2r\epsilon}$.

We go on to present algorithms solving modified $TASK_n$.  Our algorithm solving modified $COVER_n$ is called $CBoxes$, and our algorithm solving modified $NAV_n$ and modified $SEARCH_n$ is called $Boxes$.  Our algorithms are optimally competitive:

    \begin{theorem}[c.f. Theorems \ref{thm:CBoxesWorks} and \ref{thm:CUpperBound}]\label{thm:C}
    The algorithm CBoxes solves the modified $COVER_n$ problem and is optimally competitive with an upper bound on competitiveness with respect to optimal length $l_{opt}$ given by $cl_{opt}+d$, where $c$ and $d$ are constants depending on $r$, $n$, and $\epsilon$.
    \end{theorem}

    \begin{theorem}[c.f. Theorems \ref{thm:BoxesWorks} and \ref{thm:easyUpperBound}]\label{thm:D}
    The algorithm Boxes solves the modified $NAV_n$ and $SEARCH_n$ problems and is optimally competitive with an upper bound on competitiveness with respect to optimal length $l_{opt}$ given by
        $$c \frac{l_{opt}^n}{\epsilon^{n-1}}+d,$$
    where $c$ is a constant depending on $n$ and $d$ is a constant depending on $n$ and $\epsilon$.
     \end{theorem}

This paper is organized as follows.  In Section \ref{sec:competitiveness}, we define the notions of competitiveness that Gabriely and Rimon use, modifying it slightly for our purposes.  In Section \ref{sec:CBUG}, we describe the CBUG algorithm.  We modify the definition of the problems by introducing clearance parameter in Section \ref{sec:modifying}.  In Section \ref{sec:lb}, we prove Theorem \ref{thm:B} by constructing spaces realizing the given bounds.  In Section \ref{sec:Algorithms}, we define the CBoxes and Boxes algorithms, and in Sections \ref{sec:CBoxes} and \ref{sec:Boxes}, we analyze the CBoxes and Boxes algorithms, respectively, proving Theorems \ref{thm:C} and \ref{thm:D}.  In Section \ref{sec:mathmotivation}, we discuss some of the mathematical motivation underlying the algorithms in this paper.  Finally, in Section \ref{sec:improvements} we describe a number of ways of improving the execution of the various algorithms.

A computer simulation of the algorithms contained herein is available online at\\ \texttt{http://www.math.binghamton.edu/sabalka/robotmotion}.

The second author would like to thank Elon Rimon and Misha Kapovich for many interesting conversations on this material.

\section{Competitiveness} \label{sec:competitiveness}

Recall from the Introduction that Bob is a spherical mobile robot with the task of moving in an unknown environment $X$.  We want to discuss how ``good'' a particular online algorithm is for solving the given task.  To do so, we present a notion of \emph{competitiveness} for online algorithms.  The definition here is adapted from the generalized notion of competitiveness appearing in \cite{GabrielyRimon}, and allows for an arbitrary functional relationship between an algorithm's performance and the optimal performance, not just the traditional linear dependence.

Let $P$ be a task, $NAV_n$ for example.  An \emph{instance} of $P$ is a situation in which the task should be completed.  For online navigation, the instances are given by tuples $(X,S,T,r)$, with $X$ the space, $S$ the start point, $T$ the target point, and $r$ the radius.  We will denote the set of all instances for a given task by $\mathcal{I}$.
A \emph{parameter} is a function $\pi : \mathcal{I} \to \R$.  For example, define $t_{opt} : \mathcal{I} \to \R$ to be such that $t_{opt}(I)$ is the optimal time, over all algorithms, to complete instance $I$.  For any algorithm $A$ which solves $P$, we wish to bound the time required for $A$ to complete an instance by a function function of some parameter of the instance (typically, $t_{opt}$).  To that end we introduce the following definitions.

\begin{definition}
Let $A$ be an algorithm solving a task $P$ and let $\pi$ be a parameter.  Denote by $t_A$ the function which takes an instance of $P$ and outputs the total execution time for $A$ on that instance.  Define $f_{A,\pi} : \R \to \R$ to be the function given by
    $$f_{(A,\pi)}(x) = \sup_{I \in \mathcal{I}} \{t_A(I) : \pi(I) \leq x\}.$$
\end{definition}

Thus $f_{(A,\pi)}(x)$ tells us the most time $A$ could take on an instance if $\pi$ is no more than $x$.

\begin{definition}[Competitiveness]
Let $P$ be a task.  Let  $g: \R \to \R$ be a function.  We say that $g$ is a \emph{universal asymptotic lower bound on competitiveness with respect to $\pi$} if for every algorithm $A$ solving $P$, $f_{(A,\pi)} \in \Omega(g)$.  We will sometimes simply call $g$ a universal lower bound.  An algorithm $A$ solving task $P$ is \emph{$O(g)$-competitive with respect to $\pi$} if $f_{(A,\pi)} \in O(g)$.  We say that $A$ is \emph{optimally competitive} if there is $g$ such that $g$ is a universal lower bound on competitiveness and $A$ is $O(g)$-competitive.
\end{definition}

This definition of competitiveness allows for competitiveness to be quadratic, logarithmic, exponential, etc.  For example, an algorithm $A$ being (linearly) competitive in the traditional sense is equivalent to being $O(t)$-competitive, which means $t_A \leq c_1t_{opt} + c_0$ for constants $c_0$ and $c_1$.  As a linear polynomial clearly gives a universal lower bound for competitiveness, a linearly competitive algorithm is always optimally competitive.

We now turn to our motion tasks.  First, note that Bob's position uniquely determines and is uniquely determined by the coordinates of Bob's center.  We will consistently refer to Bob's position as a point via this identification.  This allows us to talk about, for instance, Bob traversing a path in $X$.  The total execution time of an algorithm $A$ solving $TASK$ may be broken up into physical travel time and onboard computation time.  We will neglect onboard computation time when measuring optimality of our algorithm.  This is a defensible assumption, as physical motion typically takes several orders of magnitude longer than onboard computation.  To simplify our analysis, we will assume that Bob always travels at a constant speed. This correlates physical travel time with the length, $l_A$, of the path Bob travels in $X$ while executing $A$, and we may replace $t_A$ with $l_A$ in our definitions above.  These simplifications allow us to compare our performance with that of an optimal offline algorithm (for which computation time is not an issue) by comparing lengths of paths.  

The optimal offline length, $l_{opt}$, is a reasonable parameter through which to discuss competitiveness.  However, we will see that for dimension $n \geq 3$ and any algorithm $A$ that solves $NAV_n$ or $SEARCH_n$, we have $l_{(A,l_{opt})}(t) = \infty$ for every $t$.  Thus, bounding path length requires more knowledge of the space than just the optimal path length.  We will modify $l_{opt}$ and $TASK$ slightly (in Section \ref{sec:modifying}) to obtain bounds on competitiveness.

Before we turn to our modification, we present what is known for the $NAV_2$ problem, which will serve as motivation for parts of our algorithms.

\section{Solving $NAV_2$:  the CBUG Algorithm} \label{sec:CBUG}

Our algorithms build on ideas from an optimally competitive algorithm for the $NAV_2$ task of navigating unknown $2$-dimensional environments, called CBUG \cite{GabrielyRimon}.  The basic CBUG algorithm is itself a refinement of a classical but non-optimally-competitive algorithm, called BUG1 \cite{LumelskyStepanov}.  In this section, we present the BUG1 and CBUG algorithms.

BUG1 is guaranteed to yield a solution - that is, Bob will move from $S$ to $T$ if possible - but has no upper bound on competitiveness.  The BUG1 algorithm is as follows:\\

\begin{center}
\fbox{
\parbox{4.4in}{
{\center {\underline{\bf BUG1$(S,T)$}} \\}
\renewcommand{\labelitemi}{\labelitemii}
\renewcommand{\labelitemiii}{\labelitemii}
  {\bf While} not at $T$:
    \begin{itemize}
    \item Move directly towards $T$.
    \item {\bf If} an obstacle is encountered:
      \begin{itemize}
      \item Explore the obstacle via clockwise circumnavigation.
      \item Move to some point $p_{min}$ on the obstacle closest
to $T$.
      \item {\bf If} Bob cannot move directly towards $T$ from
$p_{min}$:
        \begin{itemize}
        \item {\bf Return} 0; Target unreachable.

        \end{itemize}
      \end{itemize}
    \end{itemize}

  {\bf Return} 1; Target reached
}
}\\
\end{center}

BUG1 runs in time proportional to twice the entire length $l_b$ of the boundaries of (an $r$-neighborhood of) all obstacles (with an easy modification of the algorithm and slightly more careful analysis, the constants of this bound can be improved; see \cite{LumelskyStepanov}).  However, $l_b$ can be arbitrarily large, even when $l_{opt}$ is bounded. For example, consider the simple situation where $S$ and $T$ are close together, but separated by an obstacle with large perimeter (see Figure \ref{fig:BUG1}).  One advantage of BUG1 is that only a finite amount of memory is required: Bob must only remember the points $T$, $p_{min}$, and the first point encountered on the current obstacle.

\begin{figure}

\includegraphics[width=4in]{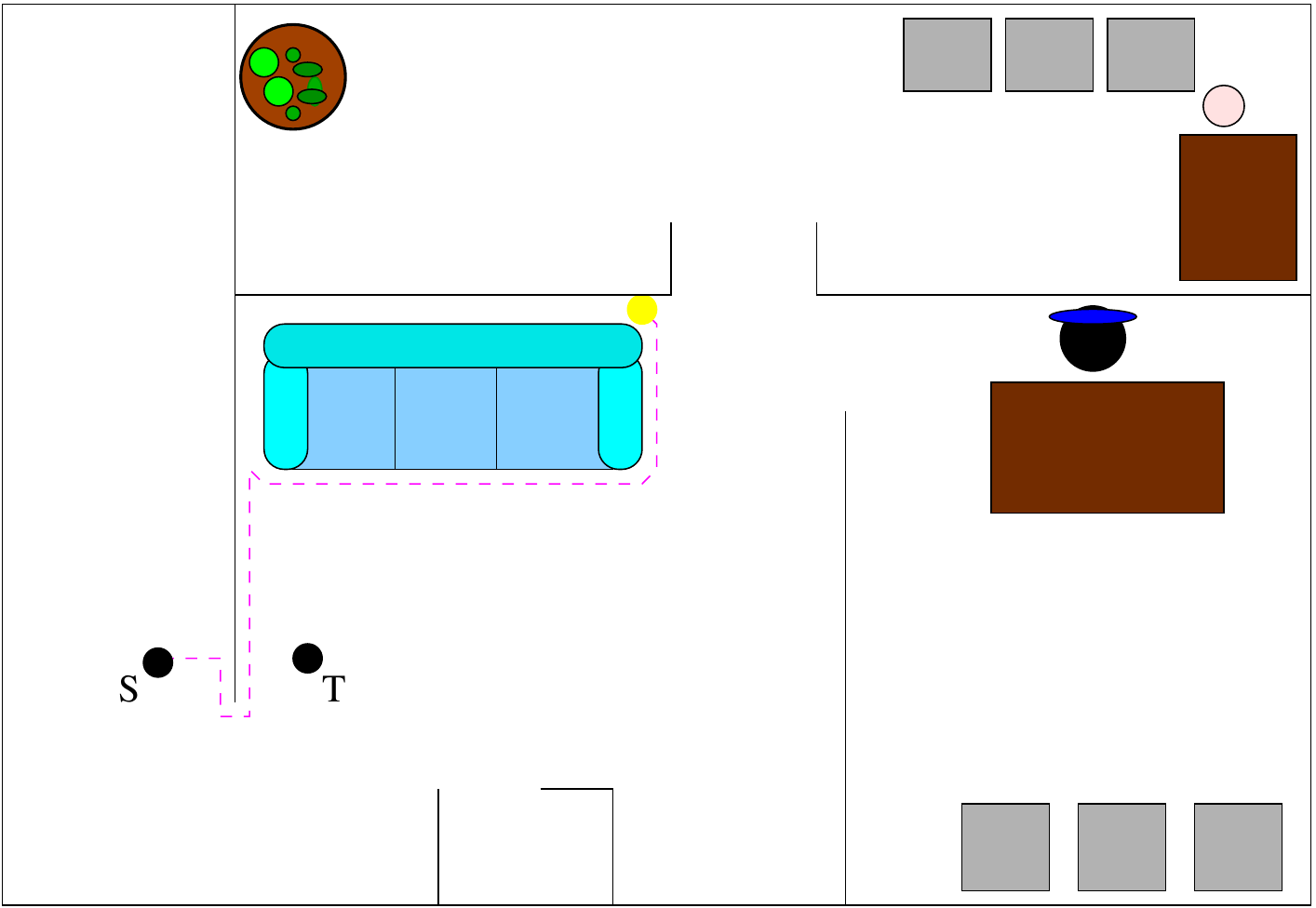}

\caption{In this instance of the $NAV_2$ problem, the BUG1 algorithm will take much longer to complete than the optimal offline solution.  Instances like this show that BUG1 is not $O(g)$-competitive for any $g$.}
\label{fig:BUG1}

\end{figure}

The CBUG algorithm solves the problem of unbounded competitiveness by introducing a virtual obstacle into the environment.  CBUG executes the BUG1 algorithm, but only within an ellipse with foci $S$ and $T$ and of fixed area $A_0$: Bob treats the ellipse as if it were an obstacle, even though it does not exist.  If BUG1 finds no solution within the given ellipse, CBUG repeats the algorithm in an ellipse of progressively larger area.  See Figure \ref{fig:CBUG}.

\begin{center}
\fbox{
\parbox{4.4in}{
{\center {\underline{\bf CBUG$(S,T,A_0)$}} \\}
\renewcommand{\labelitemi}{\labelitemii}
{\bf For} $i = 0$ to $\infty$:
  \begin{itemize}
  \item {\bf Execute} BUG1$(S,T)$ within ellipse with foci $S$ and $T$
  and area $2^iA_0$.
  \item {\bf If} Bob is at $T$:
    \begin{itemize}
    \item {\it Return} 1; Target reached.
    \end{itemize}
  \item {\bf If} Bob did not touch the ellipse while executing BUG1:
    \begin{itemize}
    \item {\it Return} 0; Target unreachable.
    \end{itemize}
  \end{itemize}
}
}\\
\end{center}

\begin{figure}
\includegraphics[width=4in]{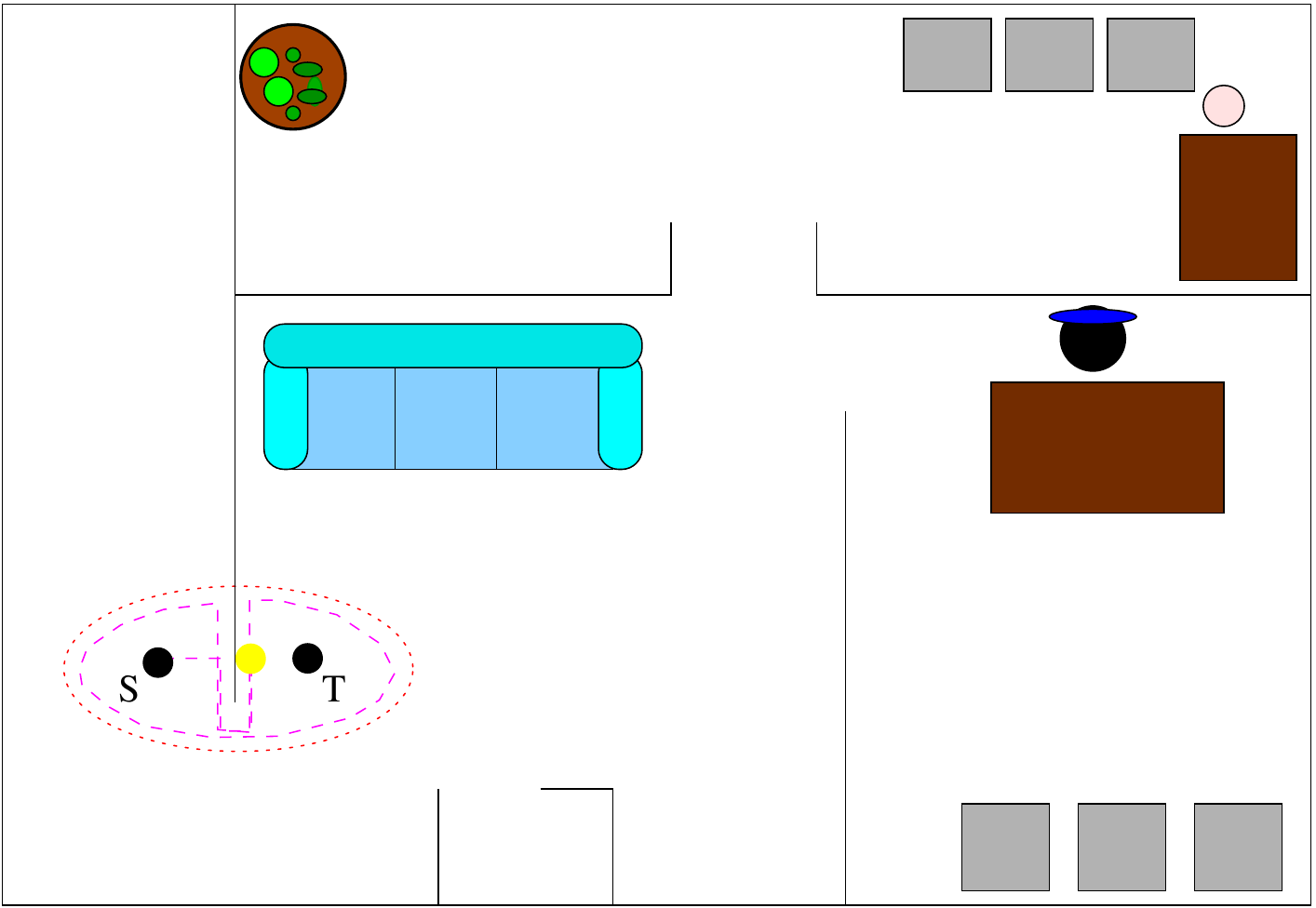}

\caption{The dashed line shows the path of a robot executing the BUG1 algorithm within a virtual bounding ellipse from the CBUG algorithm.  The ellipse prevents the robot from departing too long from the optimal path.}
\label{fig:CBUG}
\end{figure}

As the ellipses involved in CBUG are expanding in area, the virtual boundary must eventually contain either a path from $S$ to $T$ or a real obstacle cutting $T$ completely off from $S$.  In the former case, CBUG terminates at $T$.  In the latter case, Bob will not touch the virtual boundary, and again CBUG will terminate.  Note that CBUG, like BUG1, requires only constant memory: it only need remember the information $S$, $T$, $A_0$, $i$, and the point $p_{min}$ closest to $T$ on the current obstacle.  Usually, we will also have Bob remember the best path from the current point to $p_{min}$, still requiring only constant memory.

Gabriely and Rimon analyze the competitiveness of CBUG in the following two results:

\begin{theorem}\cite{GabrielyRimon}\label{thm:CBUGlower}
The $NAV_2$ problem has a quadratic universal lower bound, namely given by
    $$g_r(x) := \frac{4\pi}{6(1+\pi)^2r}x^2 \sim \frac{.122x^2}{r}.$$
\end{theorem}

In the next section we will provide a lower bound for competitiveness for $TASK$ for general $n$ (see Theorem \ref{thm:universallb}).  We note that for large $l_{opt}$, our lower bound for $NAV_2$ is tighter than the one in Theorem \ref{thm:CBUGlower}.

\begin{theorem}\cite{GabrielyRimon}\label{thm:CBUGupper}
If the target $T$ is reachable from $S$, CBUG solves $NAV_2$ in time proportional to the distance $l$, travelled by Bob, where
    $$l \leq \frac{6\pi}{2r}l_{opt}^2 + dist(S,T) + \frac{6A_0}{2r},$$
where $r$ is the robot's radius. Thus, CBUG is optimally competitive.
\end{theorem}

Further improvements in terms of constants and average-case execution can be made by slightly modifying the algorithm or allowing nonconstant memory; see \cite{GabrielyRimon} for more details.

It is not clear to the authors that the proof of Theorem \ref{thm:CBUGupper} appearing in \cite{GabrielyRimon} (Lemmas 4.1-3 and Proposition 4.4) is correct.  In particular, in the proof of Lemma 4.2 of \cite{GabrielyRimon}, the length $l$ of the path that Bob's center traverses is implicitly related to the area $A$ swept out by Bob via the formula $l\leq A/(2r)$, where $r$ is Bob's radius.  It does not seem to be apparent that this identification holds without further argument.  For instance, the $r$-neighborhood of a fractal curve has finite area, even though a fractal curve has infinite length.  What prevents the path Bob's center follows from being arbitrarily long with respect to the area of its $r$-neighborhood?  In higher dimensions, we can construct spaces which force Bob's center to travel an arbitrarily long distance while covering only a bounded volume.  For example, fix some constant $k$ and take $X \subset \R^3$ to be the $r$-neighborhood of the curve $\gamma(t) = (t,\sin(kt),0)$, with $0 \leq t \leq 1$, with $S$ and $T$ on the curve at points $\gamma(0)$ and $\gamma(1)$, respectively.  Notice that for all $k$, $X$ is a subset of the box $B = [-r,1+r] \times [r-1,r+1] \times [-r,r]$.  Given the vertical restriction, Bob's center must have height 0.  Furthermore, the only points in $X$ with height $r$ are those directly above $\gamma$.  Thus Bob's center is forced to traverse the entirety of $\gamma$.  As $k \to \infty$, the length of $\gamma$ increases without bound, but the volume swept out by Bob is bounded above by the volume of $B$.  There could be such an example for dimension $n = 2$ as well.

In Section \ref{sec:mathmotivation}, we will give an argument justifying the use of some linear relationship between $l$ and $A/(2r)$ in dimension $2$ in some cases.  Our formula will be $l \leq cA/(2r)$ for some constant $c$ (with $c$ much bigger than $1$).

\section{Modifying $TASK_n$:  Clearance Parameter}\label{sec:modifying}

We now wish to analyze the $TASK_n$ problem for arbitrary $n$.  As mentioned in the Introduction, there can be no optimally competitive algorithm for the $TASK_n$ problem, which we will prove in the next section.  However, in the process, we will find bounds on competitiveness for a slightly weaker problem, defined here.

It seems that tight corridors are problematic for online robot navigators.  One way to remove that problem is to assume that the robot has some clearance parameter.

\subsection{Notation}

We begin by introducing convenient notation to be used throughout the remainder of the paper to discuss the notion of clearance.

\begin{definition}[$\rho$-neighborhood]
Let $Y \subseteq \R^n$ and let $\rho > 0$.  Then the \emph{$\rho$-neighborhood} of $Y$ is the union of all $\rho$-balls about points in $Y$:
    $$N_\rho(Y) = \{x \in \R^n : \text{ there exists } y \in Y \text{ such that } d(x,y) < \rho \}.$$
Here $d$ measures Euclidean distance.
\end{definition}

\begin{definition}[$\rho$-path]\label{def:path}
Let $p$ be a path, and let $\rho > 0$.  If the set $N_\rho(p)$ does not intersect an obstacle of $X$ then we call $p$ a \emph{$\rho$-path} in $X$.
\end{definition}

\begin{definition}[$\kappa$ and $r'$]
Let $r$ be Bob's radius and fix a constant $\epsilon \geq 0$.  Define
    $$\kappa = 2\sqrt{2r\epsilon + \epsilon^2}$$ and
    $$r' = r+\epsilon.$$
\end{definition}

Notice that if there are two points in $X$ of distance at most $\kappa$ apart and such that a sphere of radius $r'$ can occupy either point, then Bob can move along the straight line between them.

\subsection{Clearance and the Modified $TASK_n$ Problem}

\begin{definition}[Modified $TASK$]
Fix a constant $\epsilon > 0$, called the \emph{clearance parameter}.  We define the \emph{$\epsilon$-modified} versions of each task.  For convenience, we refer to the $\epsilon$-modified version of $TASK$ as \emph{modified} $TASK$.  The \emph{modified} $NAV_n$ and $SEARCH_n$ problems are to reach $T$ from $S$ if there is an $r'$-path, and if there is not then to either reach $T$ from $S$ along an $r$-path or determine that there is no $r'$-path.  The \emph{modified} $COVER_n$ problem is to traverse a path such that Bob's center comes within $r'$ of every point that is within $r$ of an $r'$-path from $S$ (that is, Bob comes close to all points Bob can touch without getting too close to an obstacle).
\end{definition}

Note that it is possible for an $r$-path solving $TASK_n$ to exist that is much shorter than every $r'$-path.  In fact, a slight modification of the examples from section \ref{sec:lb} shows that for any algorithm $A$ solving $NAV$ or $SEARCH$, and length $l$, there are spaces with $r'$-paths from $S$ to $T$ where the optimal path length is $l$ and $A$'s path length is an arbitrarily large multiple of $l$.  Thus, it is impossible to measure competitiveness with respect to the length of the optimal $r$-path, even for modified $NAV$ or modified $SEARCH$.  For this reason, we modify $l_{opt}$:

\begin{definition}[Modified $l_{opt}$]
For each instance of modified $NAV$ or modified $SEARCH$, we modify the length $l_{opt}$ used to compute competitiveness to be the optimal offline path length for a robot of radius $r'$ instead of radius $r$.
\end{definition}

From here on, we discuss competitiveness of algorithms solving modified $TASK$ with respect to this modified parameter.

\section{Universal Lower Bounds}\label{sec:lb}

For the modified $COVER$ problem, there is an obvious linear universal lower bound.  In this section we give much stronger explicit universal lower bounds for the modified $NAV$ and $SEARCH$ problems.  Note a universal lower bound for modified $NAV$ is automatically a universal lower bound for modified $SEARCH$, as $SEARCH$ is the same problem but with less information.  Our universal lower bound will be constructed via spaces where the extra knowledge of the exact location of $T$ does not help Bob, effectively transforming an instance of $NAV$ into an instance of $SEARCH$.  Moreover, given an algorithm which solves $SEARCH$, one may always turn an instance of $SEARCH$ into an instance of $COVER$, by moving the target $T$ to the last place Bob searches.  We exploit this fact to describe spaces such that the optimal offline runtime is proportional to the side length of an $n$-cube, while an online algorithm runs in time proportional to the $n$-volume of an $n$-cube.

Let $TASK_n$ be one of $NAV_n$ or $SEARCH_n$, and suppose an algorithm $A$ solves $TASK_n$.  Let $l_0$ and $\epsilon$ be given constants.  To construct a universal lower bound, we will create a space in which there is an $r'$-path from $S$ to $T$ of length $l_{opt}$ and the path prescribed by $A$ has length on the order of $\frac{l_{opt}^n}{r'\kappa^{n-2}}$.  In particular, if $n \geq 3$ then the length of path prescribed by $A$ goes to infinity as $l_{opt}$ is fixed and $\epsilon$ goes to 0, so no algorithm is competitive with respect to $l_{opt}$ without modifying $TASK_n$.

The spaces we construct will be `parallel corridor spaces',  $PC(l_0,\epsilon,r,n)$.  Each space will consist of a number of floors, and each floor will consist of several corridors of length $l_0$. These corridors will be squeezed as closely together as possible, overlapping to a great extent but not overlapping so much that Bob can move directly from one to another.  We will create an instance of $TASK_n$ by placing $S$ and $T$ at opposite ends of the several corridors, and then blocking all but one corridor, forcing Bob to explore every corridor to get from one side to the other (that is, essentially treat the instance as one of $COVER_n$).  For an example throughout this construction, see Figures \ref{fig:parallelcorridors2D} and \ref{fig:parallelcorridors3D}.

For $n = 2$, there is a similarity between our spaces and those of \cite{BlumRaghavanSchieber}, which essentially established a universal lower bound for (a different kind of) competitiveness for the $NAV_2$ task when all obstacles are polygonal.

\subsection{Constructing the Space}

To begin constructing $PC(l_0,\epsilon,r,n)$, we first construct a finite cubical lattice $L \subset \R^{n-2}$.  We restrict the coordinates to be elements of the closed interval $[0,l_0]$.  We will carefully choose a distance, $\lambda$, between adjacent points in $L$.  Roughly, $\lambda$ is chosen to be as small as possible while still being larger than $\kappa$.  Now we will define $\lambda$ more precisely.

Denote the number of values that fit into the interval $[0,l_0]$ and spaced at least distance $d$ apart by $np(d)$.   Then $np(d) = \lfloor l_0/d \rfloor + 1$.  If $l_0/\kappa \in \N$, then $np(\kappa) = l_0/\kappa + 1$.  In this case, choose $\lambda > \kappa$ so that $np(\lambda) = l_0/\kappa$.  If $l_0/\kappa \notin \N$ then choose $\lambda > \kappa$ so that $np(\lambda) = np(\kappa) = \lfloor l_0/\kappa \rfloor + 1 > l_0/\kappa$.

Extend $L$ to a subset of $\R^{n-1}$ by attaching an interval of length $l_0$ to each point in $L$: $L \times l_0I \subset \R^{n-1}$, where $I$ is the unit interval.  The $r'$-neighborhood of each of these lines is a \emph{corridor}.  Create a series of corridors in $\R^n$ by taking the $r'$ neighborhood of $L\times l_0I$:  set $L' := N_{r'}(L\times l_0I)$.  Figure \ref{fig:parallelcorridors3D} shows a picture of $L'$ (missing caps on the ends, and with extra black `flaps') in the case $n = 3$.

Notice that $L$ could have been chosen from a more dense packing (of the $(n-2)$-cube with side length $l_0$ by spheres of radius $\lambda$) to fit in more corridors into $L'$ and thus obtain better constants for the bound (see the Lattice Improvement, Section \ref{sec:SamplingImprovement}).  However, this does not affect the competitiveness class of our example.

The set $L'$ is already a collection of parallel corridors, but there are not enough of them.  Stack $h = \left\lfloor l_0/(2r')\right\rfloor$ copies of $L'$ on top of one another.  That is, place one copy of $L'$ at height $0$, one at height $2r'$, one at height $2(2r')$, etc., up to height $h(2r')$.  Think of each copy of $L'$ as a `floor' of a building with $h$ stories.

To be able to access any corridor from any other corridor, add a room at each end of the collection of all corridors so that a robot of radius $r'$ can pass between floors by passing through a room - i.e. both rooms have dimensions roughly $(2r' \times (l_0+2r')  \times\dots\times (l_0+2r') \times (l'+2r'))$, where $l' = 2r'h$.  Call one room the start room and the other the target room.

At the end of all but one of the corridors, place obstacles that prevent passage from the corridor to the target room.  The choice of which corridor to leave open depends on the algorithm, $A$, that we are building the space for.  If all of the corridors were blocked, $A$ would visit every corridor in some order before terminating.  Leave the last one unblocked.  We should carefully choose the size and location of the obstacles so that:
\begin{enumerate}
\item they block a robot from exiting,
\item they allow a robot to exit the chosen unblocked corridor, and
\item they are small enough that a robot in one corridor cannot feel an obstacle in an adjacent second corridor (and thus determine that it need not go down the second corridor).
\end{enumerate}
Say we wish to block a corridor $C$, whose axis of symmetry is $x \times l_0I$, where $x \in \R^{n-2}$.  Let $P'$ denote the set of points in $C$ which are as close as possible to but not in the target room.  Then $P'$ is a $(n-1)$-ball orthogonal to $x \times l_0I$ in $\R^n$.  Let $P \subset P'$ denote those points whose height (i.e. the value of the last coordinate, which is the coordinate that was increased in the stacking phase, corresponding to the `floor') differs from the height of $x$ by $g$ or more, where $g = \sqrt{(r+\epsilon)^2-((\kappa+\lambda)/4)^2}$.  Then $P$ consists of two connected components of distance $2g$ apart.

We claim $P$ satisfies the three desired properties, all of which follow from the choice that $\lambda > \kappa$.  For, the connected components of $P$ are distance $2g < 2\sqrt{(r+\epsilon)^2-(\kappa/2)^2} = 2\sqrt{(r+\epsilon)^2-((r+\epsilon)^2 - r^2)} = 2r$ apart, which blocks a robot of radius $r$ from passing, so (1) is satisfied.  Let $C_1$ and $C_2$ be intersecting corridors.  Their axes of symmetry are at the same height, $z$.  Notice that the difference, in absolute value, between $z$ and the height of a point in $C_1 \cap C_2$ is at most $\sqrt{(r+\epsilon)^2-(\lambda/2)^2} < g$.  Thus the unblocked corridor has had no obstacles placed in it, so (2) is satisfied.  Furthermore, the same calculation shows that a robot can only feel elements of height strictly less than $g$ in adjacent corridors.  Thus (3) is satisfied.

Adding the obstacles $P$ to all but one corridor, we have now finished constructing the space $PC(l_0,\epsilon,r,n)$.

We create an instance of $TASK_n$ by placing the start point $S$ in the center of the start room, and placing the target point $T$ in the center of the target room.

\begin{figure}

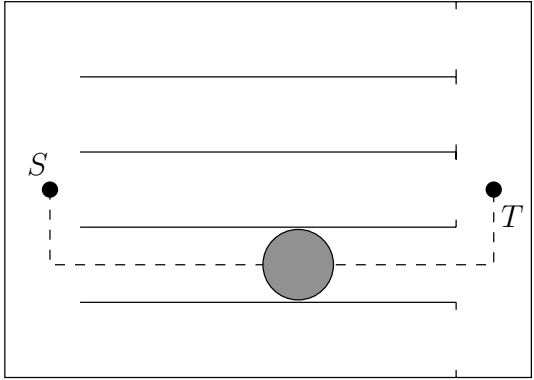

\caption{The space $PC(8r',\epsilon,r,2)$.  The circle represents the robot, and the dotted line indicates the optimal length path from $S$ to $T$.}

\label{fig:parallelcorridors2D}

\end{figure}
\begin{figure}

\includegraphics{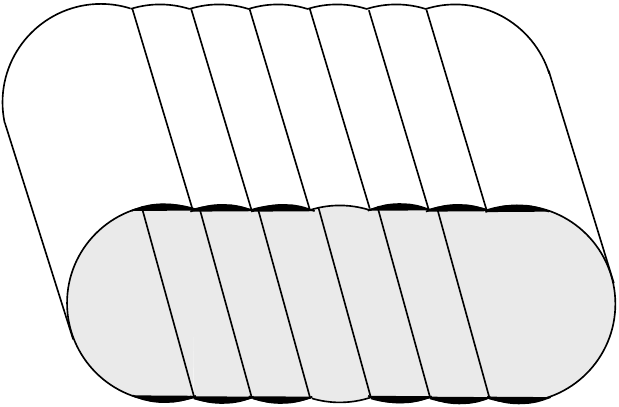}

\caption{The corridors within a set $L'$ when $n = 3$.  Notice that exactly one of them (the middle one) is unblocked; Bob cannot pass through the other corridors because of the black obstacle `flaps'.}

\label{fig:parallelcorridors3D}

\end{figure}

\begin{example}\label{ex:parallelcorridors}

Consider when $n = 2$.  In this case, $L$ is a subset of $\R^0$.  That is, $L$ is a point.  Since $L$ is a point, $L\times l_0I$ is a line segment of length $l_0$.  If we think of the line segment as sitting horizontally in the Euclidean  plane, then taking the $r'$-neighborhood of the line segment yields: two half-circles, one on the left (concave right) and one on the  right (concave left), connected by horizontal line segments of length $l_0$.  For the sake of this example, assume $l_0 = h(2r')$ for some integer $h$.  Then, taking $h$ copies of this corridor and stacking them on top of one another, we have a series of horizontal line segments spaced at distance $2r'$ from each other, with `caps' on the left and right ends.  To create the space $PC(l_0,\epsilon,r,2)$, on each end we replace the caps with a box of width $2r'$ and height $l_0$.   We declare the left-hand box to be the start box, putting $S$ in its center, and we declare the right-hand box to be the target box, putting $T$ in its center. Finally, we place obstacle `flaps' in all but one of the corridors on the far right-hand side.  See Figure \ref{fig:parallelcorridors2D}.

\end{example}

\begin{example}

If $n = 3$ then $L$ is a set of equally spaced points on a line segment of length $l_0$.  The resulting set $L'$ is as shown in Figure \ref{fig:parallelcorridors3D}, plus hemishperical `caps' on the end.  We say that two corridors are \emph{adjacent} if the distance between the corresponding points in $L$ is $\lambda$.  Notice that if $\lambda$ were allowed to be smaller than $\kappa$ then the pinching between adjacent corridors would be less and a robot of radius  $r$ could pass directly from one corridor to another without passing through the start room.

\end{example}

\subsection{Analysis of the Space}

We now begin to analyze this space.  First, we prove the following.

\begin{lemma}\label{lem:forcingbacktracks}
Fix a space $PC = PC(l_0,\epsilon,r,n)$.  To get between adjacent corridors in $PC$, Bob's center must pass through the start room.
\end{lemma}

\begin{proof}
Let $C_1$ and $C_2$ be two adjacent corridors in $PC$, with center lines $l_1$ and $l_2$.
For convenience, we will assume that $l_1$ and $l_2$ are axis parallel, and 0 on all but the first 2 coordinates.  Let $CS$ be the set of points in $C_1 \cup C_2$ where Bob's center can be.  Let $CS_i = C_i \cap CS$, where $i=1,2$.  We want to show that there is no path from $CS_1$ to $CS_2$ in $CS$.  Suppose there were.  Then since $CS_i$ is connected, and there is a path from $l_1$ to $l_2$ through $CS_1 \cup CS_2$.  In particular, Bob can be on a point, $p$, equidistant from $l_1$ and $l_2$.  WLOG, the third coordinate of $p$ is nonnegative.  But then consider the point $p_r$ you get by adding $r$ to the third coordinate of $p$.  Its distance from $l_1$ is at least $\sqrt{(\kappa/2)^2 + r^2} > r + \epsilon$, so $p' \notin CS_1$.  Similarly, $p' \notin CS_2$.  But this contradicts that Bob can occupy the point $p$.
\end{proof}

We are now ready to prove our universal lower bound.  We have been careful to keep track of the effect that $\epsilon$ has on the complexity of the problem.

\begin{theorem}\label{thm:universallb}
The modified $NAV_n$ task and the modified $SEARCH_n$ task both have an asymptotic universal lower bound on competitiveness given by
    $$l_{opt}^n.$$
Moreover, for large enough $l_{opt}$ we have that for every navigation algorithm $A$, there is a space $X$ with an $r'$-path from $S$ to $T$ of length at most $l_{opt}$ and
    $$l_A(X) \geq c_n\frac{l_{opt}^n}{\kappa^{n-2}r'},$$
where $c_n$ is a constant depending only on $n$.
\end{theorem}

Together with the additional mild assumption that $\epsilon < r$, `large enough' depends solely on $n$ and $r$.

When $n \geq 3$, by allowing $\epsilon$ to go to $0$ we obtain greater and greater lower bounds on online path length independent of $l_{opt}$.  Thus,

\begin{corollary}\label{cor:needmod}
If $n \geq 3$ then every algorithm that solves the (unmodified) $NAV_n$ problem or the (unmodified) $SEARCH_n$ problem is not $O(f)$-competitive for any $f : \R \to \R$.
\end{corollary}

\begin{proof}[Proof of Theorem \ref{thm:universallb}]

To obtain our lower bound, we examine optimal path length on parallel corridor spaces $PC(l_0,\epsilon,r,n)$.

In such a space, one way to get from $S$ to $T$ is to travel from $S$ directly to the unobstructed corridor, travel along the unobstructed corridor to the target room, then travel directly to $T$. The length $l_{opt}$ is therefore at most the length of a corridor ($l_0$) plus the maximal distance from $S$ to a corridor (which is at most
$\sqrt{n-1}l_0/2+r'$) and from a corridor to $T$ (which is also at most $\sqrt{n-1}l_0/2+r'$).  Thus,
    $$l_{opt} \leq (1+\sqrt{n-1})l_0 + 2r',$$
and so
    \begin{equation}\label{eq:l0lopt}
    l_0 \geq \frac{l_{opt} - 2r'}{1 + \sqrt{n-1}}.
    \end{equation}

The number of corridors per floor in $PC(l_0,\epsilon,r,n)$ is at least
    $$\left(\frac{l_0}{\kappa}\right)^{n-2},$$
so the total number of corridors is at least
    $$\left(\frac{l_0}{\kappa}\right)^{n-2}\left\lfloor \frac{l_0}{2r'} \right\rfloor.$$

Let $A$ be an algorithm solving $TASK_n$.  If all of the corridors in our space were blocked, $A$ would have Bob visit each corridor in some order.  Let the last-visited corridor be the one unblocked.  Then Bob travels down each blocked corridor at least twice, each time a distance of at least $(l_0 - r')$.  Bob travels the unblocked corridor at least once.  Thus, any online algorithm in $PC(l_0,\epsilon,r,n)$ will have path length at least
    \begin{equation}\label{eq:totalLength}
    2\left(\left(\frac{l_0}{\kappa}\right)^{n-2}\left\lfloor \frac{l_0}{2r'} \right\rfloor - 1\right)(l_0 - r') + l_0.
    \end{equation}

    Combining (\ref{eq:l0lopt}) and (\ref{eq:totalLength}), we have the desired result.
\end{proof}

We note the particular case $n = 2$.  Letting $\epsilon \to 0$, we see that this distance is quadratic in $l_{opt}$ with a leading coefficient of $1/(4r) = .25/r$.  Thus for large $l_{opt}$, our example forces paths more than twice as long as those forced by the examples of \cite{GabrielyRimon} (see Theorem \ref{thm:CBUGupper}).

\section{The Algorithms}\label{sec:Algorithms}

In this section, we present the CBoxes algorithm for solving the COVER problem and the Boxes algorithm for solving the $NAV$ and $SEARCH$ problems.  Both algorithms have the same structure, and rely on only a few main ingredients.  The first ingredient is to subdivide the space into a cubical lattice, discretizing the problem.  We break the space up into cubes, or `boxes', so that two points in a given cube are at most $\epsilon$ apart.  The second ingredient, for Boxes, is to restrict movement to an ellipsoid, and progressively increase the volume of the ellipsoid.  This is an important ingredient for obtaining upper bounds on complexity.  Our virtual bounding ellipsoid is a direct generalization of the virtual ellipses of Gabriely and Rimon \cite{GabrielyRimon}.  The final main ingredient is to explore unobstructed cubes by performing a depth-first search of the space\footnote{To be precise, a depth-first search of a dynamically generated spanning tree of the 1-skeleton of the dual of the cubical lattice.}.  We will analyze these algorithms in the next section.

\subsection{Colors}

To begin, we introduce some terminology to make visualization of the algorithm's execution easier and to formalize some aspects of the algorithm.  When our algorithms are being run, there are a few types of cubes that are encountered.  We describe and associate a color to each type of cube:
    \begin{itemize}
    \item White:  Unexplored;
    \item Yellow:  Bob's center can be at the center of the cube;
    \item Red:  Too close to an obstacle:  the center of a robot of radius $r+\epsilon$ cannot be anywhere in the cube;
    \item Pink:  outside of the virtual boundary.
    \end{itemize}
As our algorithms run, they change White cubes into Yellow, Red, or Pink cubes, and (when increasing the size of the virtual boundary) Pink cubes back to White cubes.

\subsection{The CBoxes Algorithm}

We are ready to define the CBoxes algorithm and its companion algorithm, CGraphTraverse.  The CBoxes algorithm is our solution to the COVER problem (hence the `C').  The CGraphTraverse companion algorithm implements the depth-first search described above.

\begin{center}
\fbox{
 \hspace{-.4in}
 \parbox{4.4in}{
    {\center {\underline{\bf CBoxes$_{\epsilon}$}} \\}
    \renewcommand{\labelitemi}{\labelitemii}
    \renewcommand{\labelitemiii}{\labelitemii}
    \begin{itemize}
        \item Break $X$ into a grid of axis-parallel cubes (`boxes') with side length
            $l = \min\{\epsilon/2,\epsilon/\sqrt{n}\}$.  All cubes begin colored White.
        \item Travel in a straight line from $S$ to the center of the current cube, $C$.  If obstacle is encountered, stop (no $r'$-paths exist).
        \item Color $C$ Yellow.
        \item Explore $X$ using CGraphTraverse($C$).
    \end{itemize}
}
}\\
\end{center}

\begin{center}
\fbox{
 \hspace{-.4in}
 \parbox{4.4in}{
    {\center {\underline{\bf CGraphTraverse($C$)}} \\}
    \renewcommand{\labelitemi}{\labelitemii}
    \renewcommand{\labelitemiii}{\labelitemii}
    \begin{itemize}
        \item Let $Adjacent$ be the set of cubes sharing an $n-1$ dimensional face with $C$.
        \item While there are White cubes in $Adjacent$,
        \begin{itemize}
            \item Pick White $D \in Adjacent$.
            \item Move in a straight line toward the center of $D$.
            \item If we encounter an obstacle in the process
            \begin{itemize}
                \item Color $D$ Red.
            \end{itemize}
            \item Else
            \begin{itemize}
                \item Color $D$ Yellow.
                \item CGraphTraverse($D$).
            \end{itemize}
                \item Travel back to the center of $C$.
        \end{itemize}
    \end{itemize}
}}\\
\end{center}

\subsection{The Boxes Algorithm}

The Boxes algorithm solves both the $SEARCH$ and $NAV$ problems, with only a minor modification between the two (for the $SEARCH$ problem, we use a virtual sphere instead of a virtual ellipsoid).  The basic algorithm does not rely on the location of $T$.  However, we will offer improvements in Section \ref{sec:improvements} which will potentially greatly improve the average-case runtime of Boxes when solving $NAV$ problems.  The improvements will come largely from the choice of the cube $D$ in the GraphTraverse algorithm.  The similarities between Boxes and CBoxes will be apparent.  As with CBoxes, the Boxes algorithm has a companion algorithm, GraphTraverse, which implements the depth-first-search portion of the algorithm.

If at any time the algorithm stops without reaching $T$, then $T$ is unreachable -- there is no $r'$-path from $S$ to $T$.

\begin{center}
\fbox{
 \hspace{-.4in}
 \parbox{4.4in}{
    {\center {\underline{\bf Boxes$_{\epsilon}$}} \\}
    \renewcommand{\labelitemi}{\labelitemii}
    \renewcommand{\labelitemiii}{\labelitemii}
    \begin{itemize}
        \item Break $X$ into a grid of axis-parallel cubes (`boxes') with side length
            $l = \min\{\epsilon/2,\epsilon/\sqrt{n}\}$.  All cubes begin colored White.
        \item Travel in a straight line from $S$ to the center of the current cube, $C$.  If an obstacle is encountered, stop.
        \item Define $T'$ to be $T$ if solving modified $NAV_n$, or $S$ if solving modified $SEARCH_n$.  Define $a_0 = d(S,T')+l$, and set $a = a_0$.
        \item While not in the same cube as $T$
            \begin{itemize}
                \item Define $\mathcal{E}$ to be the solid ellipsoid defined by $\{p : d(S,p) + d(p,T') \leq a\}$.
                \item Color Pink all cubes that are completely outside of $\mathcal{E}$.
                \item Explore $X$ using GraphTraverse($C$,$T$).
                \item If there is no cube adjacent to a Pink cube which is explored while executing GraphTraverse, stop.
                \item If $S$ is surrounded by points within Red cubes, stop.
                \item Set $a = a\times 2$.
                \item Color all Pink cubes White.
            \end{itemize}
        \item Travel in a straight line to $T$.  If an obstacle is encountered, stop.
    \end{itemize}
}
}\\
\end{center}

\begin{center}
\fbox{
 \hspace{-.4in}
 \parbox{4.4in}{
    {\center {\underline{\bf GraphTraverse($C$,$T$)}} \\}
    \renewcommand{\labelitemi}{\labelitemii}
    \renewcommand{\labelitemiii}{\labelitemii}
    \renewcommand{\labelitemiv}{\labelitemii}
    \begin{itemize}
        \item If $T \in C$ Return.
        \item Let $Adjacent$ be the set of cubes sharing an $(n-1)$-dimensional face with $C$.
        \item While there are White cubes in $Adjacent$,
        \begin{itemize}
            \item Pick White $D \in Adjacent$.
            \item Move in a straight line toward the center of $D$.
            \item If we encounter an obstacle in the process
            \begin{itemize}
                \item The obstacle cannot be virtual, so color $D$ Red.
            \end{itemize}
            \item Else
            \begin{itemize}
                \item Color $D$ Yellow.
                \item GraphTraverse($D$,$T$).
                \item If $T$ is in the current cube, Return.
            \end{itemize}
            \item Travel back to the center of $C$.
        \end{itemize}
        \item Return
    \end{itemize}
}}\\
\end{center}

\section{Analysis of CBoxes}\label{sec:CBoxes}

In this section, we prove that CBoxes works, and analyze its competitiveness.  To do so, observe that our algorithm performs a depth-first search of a particular graph $G_0$, as follows.  The space $X$ is broken into a grid of axis-parallel cubes with side length $l = \min\{\epsilon/2,\epsilon/\sqrt{n}\}$.  Let $G$ be the graph whose vertex set is the set of centers of cubes, where there is a straight edge between two centers if they share an $(n-1)$-dimensional face.  Let $\mathcal{O}$ denote the set of edges of $G$ for which there does not exist an $r'$-path between the centers of the corresponding cubes.  Let $C_0$ denote the cube containing Bob's initial position $S$.  Then $G_0$ is the connected component of $G \setminus \mathcal{O}$ containing the center of $C_0$.

\begin{theorem}\label{thm:CBoxesWorks}
The CBoxes algorithm solves the modified $COVER_n$ problem by effectively performing a depth-first search of $G_0$ along a subtree of $G$.
\end{theorem}

\begin{proof}

We first claim that the $COVER_n$ problem is solved by having Bob visit every vertex of $G_0$.

Since $l \leq \epsilon/\sqrt{n}$, the maximum distance between two points in a cube is at most $\epsilon$.  This implies that if the center of a robot of radius $r' = r+\epsilon$ can be SOMEWHERE in a cube without the robot intersecting some obstacle then Bob's center can be ANYWHERE in that same cube.  Our proof is based on this fact.

Consider a point $x \in X$ which Bob is required to come close to by modified $COVER_n$.  Then there is an $r'$-path from $S$ to a point $p$ with $d(x,p)\leq r$.  Let $C_1, C_2, \dots, C_k$ be the sequence of cubes that this path passes through, so $S$ is in cube $C_1$ and $p$ is in cube $C_k$.    Notice that $C_{i}$ and $C_{i+1}$ can be taken to share an $(n-1)$-dimensional face.  Then Bob can pass freely from the center of $C_i$ to the center of $C_{i+1}$, so the centers of $C_1$ and $C_k$ are in the same connected component $G_0$ of $G$.  Let $p_k$ denote the center of $C_k$.  Since $d(x,p) \leq r$ and $d(p,p_k)\leq \epsilon/2$, we have $d(x,p_k) \leq r+\epsilon/2 < r + \epsilon$.  Thus, every point which is required to be explored would be explored if Bob visits every vertex of $G_0$.

That CBoxes has Bob visit every vertex of $G_0$ essentially follows from the definition of the CBoxes and CGraphTraverse algorithms.  The first step of CBoxes is to move from $S$ straight to the center of $C_0$, which is possible if a robot of radius $r'$ can occupy $S$ - that is, if any $r'$-paths exist.  Then, at each subsequent step of the algorithm, Bob moves along an edge $e$ of $G$.  The only time these algorithms do not finish traversing $e$ is if in moving from a cube $C$ to an adjacent cube $D$, Bob runs into an obstacle.  In this case, CGraphTraverse($C$) has Bob return to $C$, never to visit $D$ again.  This is OK, since Bob need only visit those cubes that can contain the center of a radius $r + \epsilon$ sphere.  We'll see that in these circumstances, a sphere of radius $r + \epsilon$ cannot have its center at any point in $D$, so $e \in \mathcal{O}$:  let $d$ be a point in $D$ and let $c$ be Bob's center upon hitting an obstacle point, $o$.  Notice that the distance from $c$ to $o$ is $r$.  The distance from $c$ to the center of $D$ is at most $l$, which is at most $\epsilon/2$.  The distance from $d$ to the center of $D$ is at most $\epsilon/2$.  By the triangle inequality, the distance from $d$ to $o$ is at most $r + \epsilon$.

Finally, Bob only moves to unexplored adjacent cubes, making Bob's path a tree.  This proves the lemma.

\end{proof}

The CBoxes algorithm does solve the modified $COVER$ problem, but unfortunately does not necessarily do so competitively.  The problem comes from spaces which have a \emph{bottleneck};  we say a space has a bottleneck if there are non-obstacle points within distance $r'$ of an $r$-path from $S$ that are not within distance $r'$ of an $r'$-path from $S$.  These arise, for instance, from corridors of diameter between $2r$ and $2r'$.  If there is a bottleneck, a radius $r$ robot might, for example, start in a very small room and unwittingly travel through a corridor of radius less than $2r'$ into a very big room.  Modified $COVER$ only requires coverage of the small room, so a robot that covers the large room is not optimal.  We see two ways of alleviating this situation.  The first is to provide Bob with a myopic visual sensor, able to detect bottlenecks:  that is, able to detect all obstacle points within distance $\epsilon$ or so.  To keep with the non-visual emphasis of this paper, we choose to analyze the second solution, by restricting our spaces to have no bottlenecks.

\begin{theorem}\label{thm:CUpperBound}
Assume $\epsilon < 2r$.  For a given space $X$ without bottlenecks, let $l_{opt}$ be the length of an optimal path solving modified $COVER_n$.  Then there exist constants $c(n,r,\epsilon)$ and $d(n,r,\epsilon)$ such that the length of the path generated by CBoxes$_{\epsilon}$ is at most $c(n,r,\epsilon)l_{opt}+d(n,r,\epsilon).$
\end{theorem}

For fixed $n$, $r$, and $\epsilon$, this is a linear upper bound, making CBoxes$_{\epsilon}$ optimally competitive:

\begin{corollary}\label{cor:Eoptimallycompetitive}
When restricted to spaces without bottlenecks, the CBoxes$_{\epsilon}$ algorithm is optimally competitive for solving the modified $COVER_n$ problem.
\end{corollary}

\begin{proof}[Proof of Theorem \ref{thm:CUpperBound}]

By Lemma \ref{thm:CBoxesWorks}, most of the CBoxes algorithm is in a depth-first search of a subtree of $G$.  The number of edges in a tree is the number of vertices minus 1, and Bob travels on each edge exactly twice.  The length of each edge is $l$.  Let $o$ be the optimal $r'$-path that solves $COVER$.

We claim that the number, $c$, of cubes that $o$ passes through is at most $3^n(l_{opt}/l+1)$.  To see this, first consider a path of length $l$ (the width of a box).  The number of cubes that this path passes through is at most the maximum number of cubes that intersect an $l$-ball.  Projecting this $l$-ball onto any dimension, we see that it intersects at most 3 cubes (in that dimension).  Thus the ball is bounded by a bounding box, three cubes on a side, so a path of length $l$ intersects at most $3^n$ cubes (In fact, the actual maximum is $3 \times 2^{n-1}$).  Now, break $o$ into several paths of length $l$ and one of length at most $l$.  If $l_{opt} = 0$, there is one such sub-path, and otherwise there are $\lceil l_{opt}/l \rceil$ sub-paths.  At any rate, the number of sub paths is at most $l_{opt}/l+1$, and each one intersects at most $3^n$ cubes, proving the claim.

Now, we claim that the number of cubes whose centers CBoxes visits or tries to visit is $O(c)$.  Let $C$ be such a cube.  Then Bob's center comes within $l < \epsilon < r'$ of the center of $C$.  Since there are no bottlenecks, $C$'s center is within $r'$ of a point $p$ on an $r'$-path.  Since $o$ solves $COVER$, $o$ contains a point $q$ within $r'$ of $p$.  Let $D$ be the cube containing $q$.  By the triangle inequality, $C$ is contained in the radius $\epsilon/2 + l + r' + r' + \epsilon/2 \leq 3.5\epsilon + 2r$ ball centered at $D$.  If $V(3.5\epsilon + 2r)$ is the volume of the $n$-sphere of radius $3.5\epsilon + 2r$, then there are at most $V(3.5\epsilon + 2r)/l^n$ cubes in this sphere.  This is certainly no more than $[(3.5\epsilon + 2r)/l]^n$. Hence, the number of cubes visited successfully or unsuccessfully by Bob is at most $c[(3.5\epsilon + 2r)/l]^n$.

The length of Bob's path is at most $2l$ times the number of Boxes Bob visits or tries to visit.  This is at most $2lc[(3.5\epsilon + 2r)/l]^n \leq 2l3^n(l_{opt}/l+1)[(3.5\epsilon + 2r)/l]^n$.

\end{proof}

\section{Analysis of Boxes}\label{sec:Boxes}

\begin{theorem}\label{thm:BoxesWorks}
If there is an $(r+\epsilon)$-path, $p$, from $S$ to $T$, then Boxes$_{\epsilon}$($S,T$) will move Bob from $S$ to $T$.
\end{theorem}

\begin{proof}
    If there is an $r'$-path from $S$ to $T$ then the analysis from Theorem \ref{thm:CBoxesWorks} guarantees that there is a path moving directly between centers of adjacent cubes from $S$ to $T$.  Given a large enough bounding ellipsoid, Bob will find this path.
\end{proof}

We now compute an upper bound on complexity for Boxes.  First, though, note that Gabriely and Rimon use ellipses as the virtual boundary obstacles.  In fact, for the $NAV$ problem, the (rotationally symmetric) ellipsoid is the optimal shape in general, as an ellipsoid is precisely the locus of points along paths from $S$ to $T$ of a given length, which will play a roll in the proof below.

 \begin{theorem}
     \label{thm:easyUpperBound}
Let $l_{opt}$ be the length of the optimal path from $S$ to $T$ for a robot of radius $r + \epsilon$.   Then the length of the path generated by Boxes$_{\epsilon}$ is at most
    $$c_n (l_{opt})^n {\left(\frac{1}{\epsilon}\right)}^{n-1}+\frac{d_n}{\epsilon}+\epsilon$$
where $c_n = \frac{16 \cdot 32^{n-1}}{2^n-1}$ for $n = 2, 3$ and $c_n = \frac{16^nn^{(n-1)/2}}{2^n-1}$ for $n > 3$, and $d_n = 2*6^n$ for $n = 2,3$ and $d_n = \sqrt{n}*6^n$ for $n > 3$.
\end{theorem}

\begin{corollary}\label{cor:optimallycompetitive}
For a fixed $\epsilon$, the Boxes$_{\epsilon}$ algorithm is optimally competitive when solving both the modified $SEARCH_n$ and modified $NAV_n$ problems.
\end{corollary}

\begin{proof}
As is the case with the CBoxes algorithm, the Boxes algorithm has Bob traverse a subtree of the graph $G$ defined in Section \ref{sec:CBoxes}.  The number of edges in a tree is the number of vertices minus 1, and Bob travels on each edge at most twice.  Furthermore, when Bob makes a false start down a blocked edge to some cube, Bob never attempts to move to that cube again (as it will be colored Red).  Thus the path of Boxes during any iteration is at most $2l$ times the number of cubes that intersect or are contained in $\mathcal{E}$.  For an easy and succinct upper bound on the number of such cubes, we note that the ellipsoid is contained in the hypercube of volume $(2a)^n$ centered at the barycenter of the ellipsoid.  This hypercube intersects at most $(2a/l + 2)^n < (4a/l)^n$ cubes in the partition of $X$.

Suppose Boxes has terminated after iteration $i$, with $a=2^ia_0$.  If $i > 0$, then $a$ is finally large enough that GraphTraverse finds its way to $T$, while the previous bounding ellipsoid is not large enough.  In particular, $l_{opt} > a/2 = 2^{i-1}a_0$, since otherwise the optimal path would be entirely within the previous bounding ellipsoid, and hence Boxes would have found a path between the centers of White cubes in that iteration.  Thus the total number of edges that Boxes traverses is at most

        \begin{eqnarray*}
        \sum_{j=0}^i \left(\frac{4\cdot 2^{j}a_0}{l}\right)^n
        &=&\left(\frac{4a_0}{l}\right)^n\sum_{j=0}^{i}2^{jn}\\
        &=&\left(\frac{4a_0}{l}\right)^n\frac{(2^{i + 1})^n-1}{2^n-1}\\
        &<&\left(\frac{4}{l}\right)^n\frac{(2^{i + 1}a_0)^n}{2^n-1}\\
        &<&\frac{1}{2^n-1}\left(\frac{16l_{opt}}{l}\right)^n.
        \end{eqnarray*}
If $i = 0$, then of course $l_{opt} \geq d(S,T)$, so we still have $l_{opt} > a_0/2$ unless possibly when $d(S,T) \leq l$.  If $d(S,T) \leq l$, $a = a_0 \leq 2l$.  In this case, the total number of edges Boxes traverses is at most
        $$\left(\frac{2\cdot a}{l}+2\right)^n \leq 6^n.$$
Each edge has length $l$.  Remembering that we have to move to and from the centers of the first and last cubes, we may need to travel an additional length $\sqrt{n}l \leq \epsilon$.  Noting that $l$ is defined in terms of $\epsilon$ gives an upper bound on the total distance traveled while executing Boxes:
    $$\frac{l}{2^n-1}\left(\frac{16l_{opt}}{l}\right)^n + 6^nl + \epsilon \leq c_n (l_{opt})^n {\left(\frac{1}{\epsilon}
    \right)}^{n-1}+\frac{d_n}{\epsilon}+\epsilon,$$
    where $c_n$ and $d_n$ are the constants in the statement of the theorem.

\end{proof}

\section{Motivations and Observations}\label{sec:mathmotivation}

We wish to comment on the motivations and observations for the various algorithms above.

The key idea for CBoxes and Boxes is to consider the tasks from a coarse-geometric viewpoint.  The introduction of the clearance parameter $\epsilon$ and the modification of the tasks allow us to approximately, or coarsely, achieve the initial goals.  The introduction of the grid of cubes is to have Bob discretely sample the unknown environment.  Indeed, the ability to navigate around an obstacle in dimension greater than 3 cannot be accomplished by simple clockwise traversal of an object -- a robot cannot touch every point on the boundary of an $n>1$ dimensional obstacle in finite time -- so some discretization is necessary (as is having a nonconstant amount of memory).

In terms of obstacle boundaries, the discretization we invoke should be familiar to, for instance, image analysts, graphics programmers, or anyone modelling 3-dimensional objects.  We use the discrete sampling of the space to form a mesh of points near an object, and approximate the object using this mesh.  In mathematical terms, we are essentially taking the shadow of the $\epsilon$-Rips complex of centers of cubes which are within $r$ of an obstacle.  The \emph{$d$-Rips complex} of a set of points $V$ with known distances between the points is the abstract simplicial complex such that there is a simplex with vertices $\{v_\alpha\} \subset V$ if and only if the maximal distance between points in $\{v_\alpha\}$ is at most $d$.  The \emph{shadow} of the Rips complex is the projection of the Rips complex into $\R^n$, where a simplex with vertices $\{v_{\alpha}\}$ is mapped to the convex hull of the points $\{v_{\alpha}\}$.  Thus, the shadow of the Rips complex can be thought of as a \emph{local convex hull} of grid points near obstacles.  Although Bob does not actually compute the shadow of the Rips complex when executing our algorithms, this was a motivation for our algorithm.

In fact, we believe it is possible to eliminate dependence on a particular decomposition of $X$ into cubes, allowing motion in arbitrary directions and allowing arbitrary points of contact with obstacles.  Such a modification could require computation of the Rips complex and its shadow.

\begin{figure}
\includegraphics[width=1.5in]{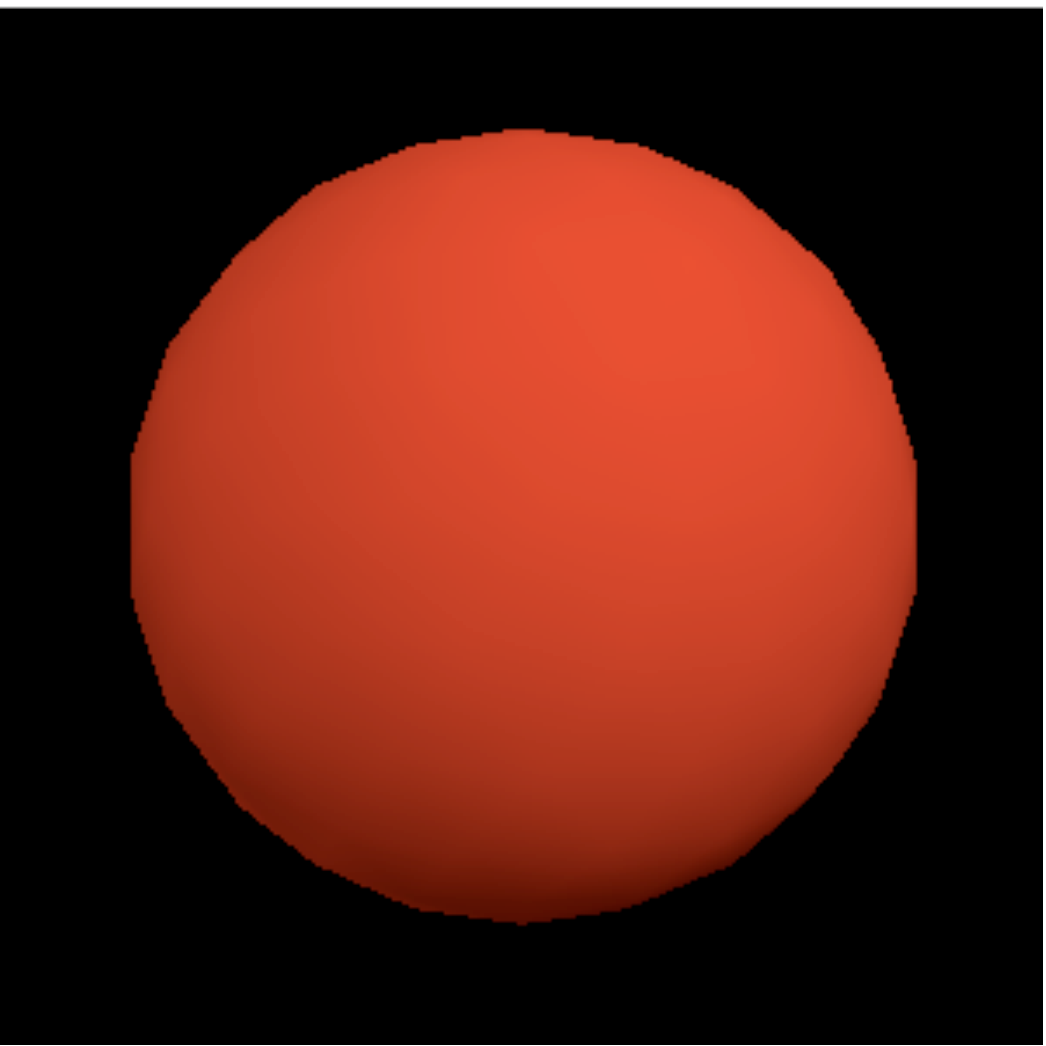}
\includegraphics[width=1.5in]{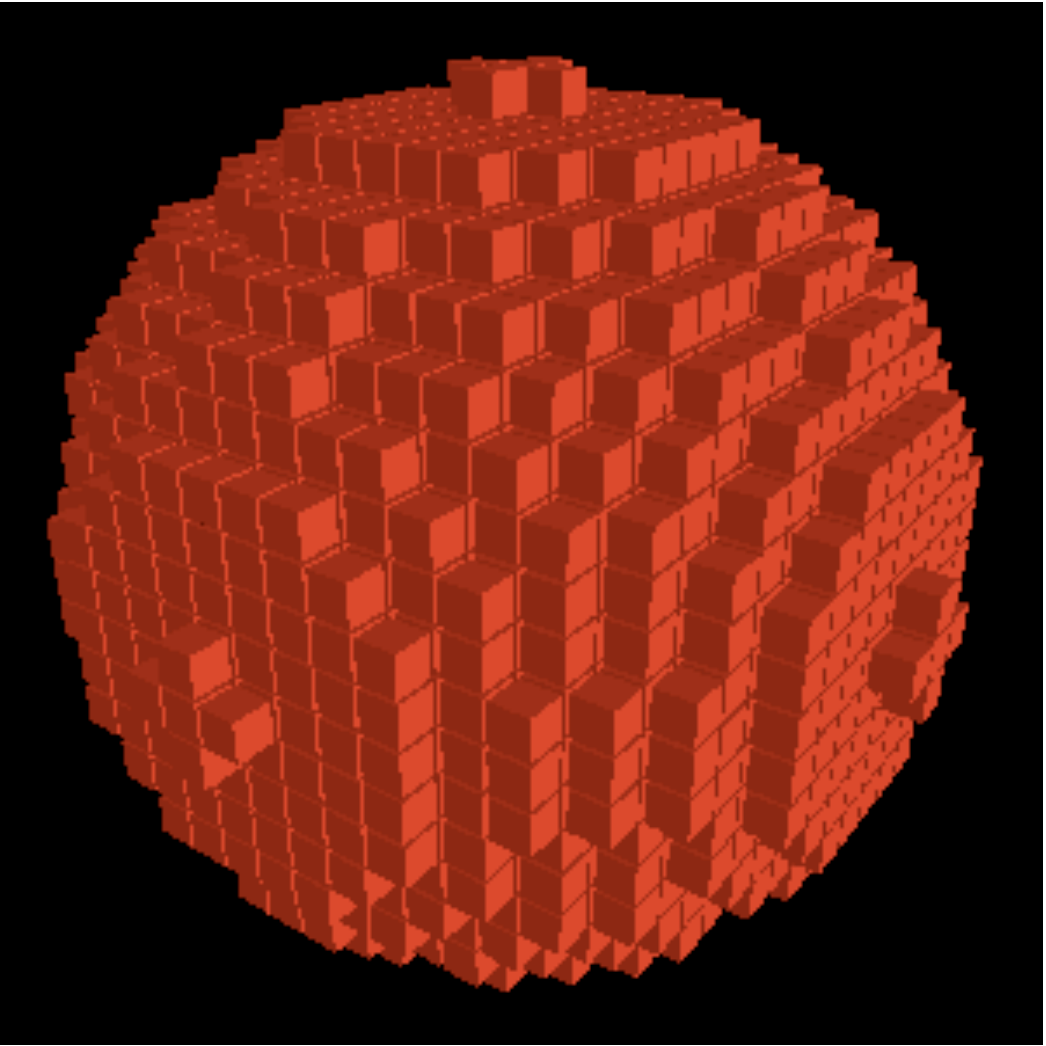}
\includegraphics[width=1.5in]{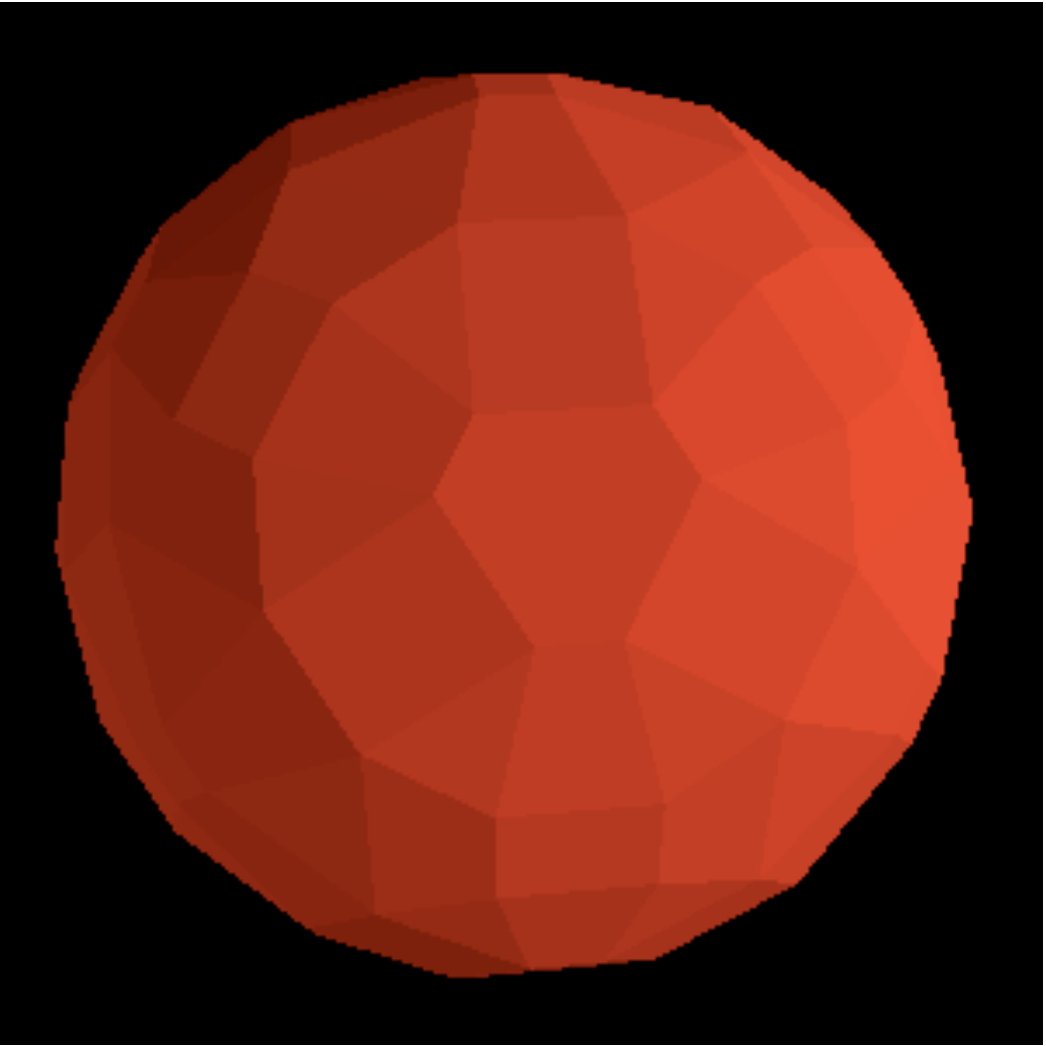}
\caption{On the left is a spherical obstacle.  The central figure shows which cubes in the given cube decomposition of the surrounding space intersect the spherical obstacle.  On the right is the shadow of the appropriate Rips complex, which is essentially the approximation of the obstacle our algorithm uses.}
\end{figure}

We also wish to mention an interesting result of Caraballo related to navigation problems.  Caraballo's result states in the case of $\R^n$ that:

\begin{theorem}\cite{Caraballo}
Let $C$ be a compact subset of $\R^n$.  For any point $q \in X$ and for almost every $r > 0$:
    $$Vol_{n-1}((d_C^{-1}(r)) \cap B^n(q, 2r)) \leq 4^{n+1}r^{n-1},$$
where $d_C(x) := d(x,C)$ and $B^n(q,2r)$ is the $n$-ball of radius $2r$ about $q$.
\end{theorem}

Let us reinterpret Caraballo's result:

\begin{corollary}
Let $X \subset \R^n$ be a space to be explored.  Let $\Delta$ denote the boundary of the set of all points in $X$ of distance at most $r$ from an obstacle point of $X$.  For almost all $r$,
\begin{enumerate}
\item The volume of $\Delta$ is finite.
\item There is a uniform bound $K(r)$ on the volume of $\Delta \cap B(q,2r)$ for any point $q \in X$.
\end{enumerate}
\end{corollary}

For an example of a surprising implication of Caraballo's Theorem, consider a fractal curve $C \subset \R^2$.  Let $X$ consist of all points within $r+\epsilon$ of $C$, and let $\Delta$ denote the boundary of the set of all points in $X$ of distance at most $r$ from its boundary (that is, points of distance $\epsilon$ from $C$).  If $\epsilon = 0$, then $\Delta = C$ and $\Delta$ has infinite length.  But Caraballo's theorem says that for almost all other $\epsilon$, the length of $\Delta$ is finite!

The importance of Caraballo's Theorem is in the uniform control on how much a robot is meant to explore given a space $X$.  Although this control on volume proves nothing about an algorithm's runtime in higher dimensions, it suggests that our tasks are at least close to solvable.

In environments with finitely many obstacle points, Caraballo's Theorem has more directly applicable consequences.  We claim that, in such environments and for $n = 2$ in particular, Caraballo's Theorem justifies the quadratic competitiveness of the CBUG algorithm (see Theorem \ref{thm:CBUGupper} and the following discussion), as follows.

\begin{corollary}
Fix $n = 2$ and an environment $X$ with finitely many obstacle points.  For every $r > 0$, there exists some constant $k=k(r)$ such that, if Bob traverses some path $\gamma$ of finite length $l$ by moving along the boundary of an obstacle, then $l$ is at most $k\cdot A/(2r)$, where $A$ is the area swept out by Bob while traversing $\gamma$.
\end{corollary}

\begin{proof}

Consider the function $f: \R \to \R \cup \{\infty\}$ defined so that $f(t)$ is the $1$-volume of the boundary of the $t$-neighborhood of the obstacles of $X$.  By Caraballo's result, for almost all $t$, $f(t)$ is finite with an explicit upper bound.  As there are finitely many obstacle points in $X$, it is a small exercise to see that $f$ varies continuously with $t$.  Thus, Caraballo's result shows that $f(t)$ is everywhere finite with an explicit upper bound.  If we cover $\gamma$ with $N$ balls of radius $r/2$, then Caraballo's result tells us that, since balls of radius $r/2$ are contained in balls of radius $2r$,
    $$l \leq N\cdot K(r).$$

We now choose a particular covering of $\gamma$ by balls of radius $r/2$.  Take the lattice of points in $\R^2$ such that each coordinate of each point is an integer multiple of $r/(2\sqrt{2})$.  Note the maximal distance from any point in $\R^2$ to a lattice point is $r/4$.  Place a ball of radius $r/2$ about each lattice point, so that every point in $\R^2$ is in some ball.  Now, keep only those balls which intersect $\gamma$.  Let $N$ denote the number of such balls.

Each of the balls in our chosen covering has diameter $r$ and contains a point of $\gamma$, so is completely contained in the $r$-neighborhood of $\gamma$.  Also, any given point of the $r$-neighborhood of $\gamma$ is contained in at most $4$ balls, by the choices made in placing the balls.  Thus, dividing the sum of the volumes of the $N$ balls by $4$ gives a lower bound on $A$:
    $$N\cdot Vol_2(B(r/2))/4 \leq A.$$

Combining the two above calculations, we obtain:
    $$\frac{l}{2^6r} = \frac{l}{K(r)} \leq N \leq \frac{4}{Vol_2(B(r/2))}A = \frac{2^4}{\pi r^2}A.$$
Solving for $l$, we have:
    $$l \leq \frac{2048}{\pi}\frac{1}{2r}A.$$
This provides the desired $k$.

\end{proof}

Note Bob always follows a $1$-dimensional path, so Caraballo's result will not help us estimate lengths of paths when $n > 2$, even for nice values of $r$ and finitely many obstacle points.  However, we note there are statements of similar results for higher values of $n$ and finitely many obstacle points, estimating the $(n-1)$-volume of the boundary of the neighborhood of obstacle points.

For finite numbers of obstacle points, $\Delta$ consists of arcs of circles - i.e. $\Delta$ is a smooth curve along which a robot can roll.  The Caraballo result says that $\Delta$ has finite length, but it is important to note this does not take into account time taken for a robot to change directions.  That is, it is assumed that a robot can turn instantaneously.  If turning time is taken into account, then there is probably no upper bound on runtime for an algorithm tracing $\Delta$.

\section{Observations and Improvements}\label{sec:improvements}

We are able to make a number of improvements to the algorithms described.

\subsection{Sampling Improvement for $COVER$, $SEARCH$, and $NAV$}\label{sec:SamplingImprovement}
For this paper, we have chosen to discretely sample the unknown environment $X$ via the centers of a grid of cubes.  These centers form a lattice (in fact, a cubical lattice:  they are the vertices of the dual cubical tiling).  Let the \emph{diameter} of a lattice denote the maximum distance between points in a primitive cell of the lattice -- that is, a fundamental domain of the quotient of $\R^n$ by the translational symmetries of the lattice.   The only mathematical properties of the lattice we used were that every point in $\R^n$ was within $r$ of a point of the lattice, and that the diameter was at most $\epsilon$.  In fact, other lattices would work.  One should be able to choose a more efficient lattice structure to sample the space with fewer lattice points.  This problem is closely related to that of sphere-packing.  Duals of lattices associated to optimal sphere-packing seem to reduce the number of lattice points per volume needed.  In particular, using the duals of the lattices associated with Gauss's hexagonal sphere-packing in dimension 2 or close packings in dimension 3 should yield better results in those dimensions.  Indeed, if one could find a good way of encoding it, even a good irregular sphere-packing would yield a better sampling of $X$.

\subsection{Taking Diagonals Improvement for $COVER$, $SEARCH$ and $NAV$}
Our complexity estimates in part relied on the distances between centers cubes sharing a codimension-1 face.  However, one can obtain similar estimates even if one allows Bob to travel from the center of one cube to the center of any other adjacent cube, sharing a face of arbitrary codimension.  This may worsen the complexity estimates, but should improve average-case runtime by a factor of up to $\sqrt{n}$.

\subsection{Noticing $T$ Improvement for $SEARCH$}
While trying to solve the $SEARCH_n$ problem, it will occasionally happen that Bob finds out where $T$ is but cannot move its center directly to $T$ because of nearby obstacles (for instance, when $\epsilon < (\sqrt{2}-1)r$ and $T$ is close to the center of a gap in obstacles slightly smaller than Bob).  Whenever $T$ is discovered, the Boxes algorithm should begin to treat $SEARCH_n$ as if it were a $NAV_n$ problem, and use the improvements below for choosing the cube $D$ referenced in the GraphTraverse algorithm and travelling expediently to $T$.

\subsection{Maximal Coloring Improvement for $COVER$, $SEARCH$ and $NAV$}
One straightforward improvement to the algorithms is to take full advantage of knowing a point on the boundary of an obstacle.  Currently, if Bob runs into an obstacle, only the cube $D$ that Bob was trying to get to is colored Red.  But Bob knows that many other cubes should also be colored Red.  The Maximal Coloring Improvement is, whenever an obstacle point is encountered, to color all cubes Red that have all corner points within distance $r'$ of the given obstacle point.  As a cube is convex, this is equivalent to saying a robot of radius $r'$ with center in the cube will intersect the obstacle point.

If we are solving $SEARCH$ or $NAV$ and we know $T$ is in a Red cube, stop.  $T$ cannot be reached.

This improvement will cause many White and Pink cubes to be colored Red, and occasionally will cause a Yellow cube to be colored Red.  To take this into account, Bob needs to check and see if the Yellow cube $C'$ colored Red is directly between the cube $C_S$ containing $S$ and the current cube $C$ in the spanning tree generated by the CBoxes and Boxes algorithms.  If $C'$ is between $C_S$ and $C$, then Bob should immediately return to the cube before $C'$, ignoring any White neighbors of cubes between $C'$ and $C$.  Either these neighbors will be explored via some other route, or they are not reachable by a robot of radius $r'$ and so should not be explored.

\subsection{Gray Improvement for $NAV$}
For the modified $NAV_n$ problem, another improvement may be made by adding a new color.  As written, GraphTraverse will explore every possible White cube, even if exploration would give Bob no new information on how to get to $T$.  For instance, consider a space with a very large sphere about $S$ as an obstacle separating $S$ from $T$.  Place a hole in the sphere so that a robot of radius $r+\epsilon$ can fit through.  Imagine that Bob has explored the entire inner boundary of the sphere, and finally reaches the hole.  Clearly, Bob should exit the sphere, as exploring any more boxes inside the sphere would just require backtracking, and Bob knows it.  This knowledge should be incorporated into the algorithm, and can be as follows.  We create a new color designation:
    \begin{itemize}
    \item Gray:  Never to be explored.
    \end{itemize}
Bob doesn't know what's in a Gray cube, but Bob will never go into one.  If there is ever a connected component $Z$ of the union of all White cubes that doesn't contain $T$, color every cube in $Z$ Gray.  Any path through centers of cubes to $T$ through $Z$ can be replaced by a path not through $Z$ (eventually, entirely through Yellow cubes).

When combined with the Pink color designation for the Boxes algorithm, note that which cubes are White and which are Gray should be recomputed by Boxes between executions of GraphTraverse.

\subsection{Greedy Improvement for $NAV$}
Coloring cubes Gray can potentially save Bob unnecessary exploration time by helping decide which cube $D$ to explore next while executing the GraphTraverse algorithm.  In fact, there is an even more efficient way of choosing which cube $D$ to explore in the GraphTraverse algorithm.  At every step, choose $D$ as follows.  If there is a path from the current cube $C$ through centers of cubes to $T$ such that all cubes on the path are colored White except possibly at the endpoints, then find a shortest such path $\gamma$.  Choose $D$ to be the next cube along $\gamma$ from $C$.  If there is no path from $C$ though centers of White cubes to $T$, there is no need to explore any adjacent unexplored cubes.  In fact, with the Gray Improvement, there can be no adjacent White cubes:  all adjacent unexplored cubes can be colored Gray.  In this situation, GraphTraverse will have Bob back up to the last cube which is not surrounded by non-White cubes.  If ever there does not exist a path from a previously explored Yellow cube through White cubes to $T$, stop:  no path exists from $S$ to $T$ for a robot of radius $r+\epsilon$.  In other words, choose $D$ greedily, and this is guaranteed to work.  Note this way of choosing $D$ does not actually require the introduction of the color Gray, and this improvement supersedes the Gray Improvement.

We note that this improvement can in particular be applied to $Boxes_{\epsilon}$ when in $2$-dimensional environments.  Compared to CBUG, $Boxes$ has two drawbacks:  the requirement of nonconstant memory, and the introduction of the clearance parameter $\epsilon$, particularly in the dependence on $\epsilon$ in the upper bound on competitiveness.  However, both algorithms are optimally competitive with respect to modified $l_{opt}$, and in many environments the Greedy Improvement will help $Boxes$ by always proceeding towards the target instead of exploring the entirety of an obstacle.

\subsection{Wide Open Spaces Improvement for $COVER$, $SEARCH$ and $NAV$}

Currently, our algorithms use very small cubes to explore $X$.  If $X$ has a large open area to explore, this can be wasteful.  Just as the Maximal Coloring Improvement takes advantage of where obstacles are, we should also take advantage of where obstacles are not.  We can do this with the following observation.  Let $N$ denote the $r'$-neighborhood of the center of a White cube $C$.  If $N$ is contained in the union of $r$-neighborhoods of all (nearby) centers of Yellow cubes, then we know even without visiting $C$ that $C$ should be colored Yellow (or some color designating that the cube need not be visited).

If the cubes which are visited are chosen carefully, this can greatly reduce the number of cubes which need to be explored.  We note, however, that this improvement is mostly unnecessary when using the following Subdivision Improvement, which is similar in essence.

\subsection{Subdivision Improvement for $COVER$, $SEARCH$ and $NAV$}

Our algorithm as stated subdivides the ambient space into cubes which are as small as necessary to prove our theorems.  But boxes of side length less than $\epsilon/\sqrt{n}$ can be too small in large, sparsely obstructed environments.  The only time it was necessary for us to use such small cubes was when proving our algorithm successfully executes in task instances which require Bob to pass through tight corridors, of diameter greater than $r+\epsilon$ but not by much.  We may search for such tight spaces using a much coarser exploration grid (i.e. much larger boxes), and only subdivide one of these larger boxes into smaller boxes when necessary.  

We present here the precise subdivision algorithm for $NAV$.  The other algorithms are similar.

Begin by breaking $X$ into a grid of axis-parallel cubes of side length $l'$ on the order of $r$, so that the centers of adjacent cubes (sharing a face with arbitrary codimension) are no more than $2r$ apart -- say, $l' := r/(2\sqrt{n})$.  Start with a fixed-radius bounding ellipsoid and execute Boxes with the following changes.  When an obstacle point, $p$, is encountered, color cubes Red like in the Maximal Coloring Improvement.  Subdivide a non-Red cube into $3^n$ subcubes if its side length is greater than $l$ and it intersects the $r'$ ball centered at $p$.  Color each of the newly created cubes as appropriate:  White by default, Red if all corners are within $r'$ of the obstacle point, Pink if entirely outside of the virtual bounding ellipsoid, and Yellow when the center has previously been visited.  Subdivided cubes are adjacent to any cube with which they share any portion of a face.  If execution stops without reaching $T$, color every Yellow cube White and restart with the new set of cubes.  Continue until no new subdivisions are created.  If at this point $T$ has not been reached, expand the bounding ellipsoid and repeat.
    
    \begin{proof}[Proof of correctness]  Suppose the ellipsoid is large enough to contain an $r'$ path, $\rho$, from $S$ to $T$.  Then we claim that the algorithm finds $T$ before expanding the ellipsoid again.  Suppose not.  Then the algorithm goes through an iteration without finding $T$ and without subdividing a cube.  Since $\rho$ is an $r'$ path, it never enters a Red cube.  Thus $\rho$ either never leaves Yellow cubes or it  enters a White cube, $W$, adjacent to a Yellow cube, $Y$.  
    
    In the first case, Bob visited the center of the cube $C$ containing $T$ but was unable to move straight from the center to $T$, so Bob must have hit an obstacle point, $p$.  The $r'$ neighborhood of $p$  contains Bob's center and thus intersects $C$.  Since $\rho$ also intersects $C$, $C$ is not entirely contained in the $r'$ neighborhood of $p$.  Thus $C$ should have been subdivided, and we have reached a contradiction.
    
    In the second case, Bob attempted to move from the center of  $Y$ to the center of $W$ but encountered an obstacle point.  Similar to the reasoning in the previous paragraph, either $Y$ or $W$ should have been subdivided.
    \end{proof}

The Subdivision Improvement means that, for the vast majority of the time, our algorithms will quickly move about in large steps.  Although we give no analysis here, the length of the path travelled using the Subdivision Improvement can be made to be at most a constant times the length of the path travelled without the improvement by limiting the number of times the algorithm can restart within each ellipsoid.  Thus, the upper bound on competitiveness with the Subdivision Improvement is in the same complexity class as the basic algorithm.  In typical cases, with the Subdivision Improvement our algorithms should not only run faster but need far less memory to execute.

\bibliographystyle{plain}
\bibliography{refs-BKS1}

\def\cprime{$'$}
\begin{thebibliography}{10}

\bibitem{BaezaYatesCulbersonRawlins}
Ricardo~A. Baeza-Yates, Joseph~C. Culberson, and Gregory J.~E. Rawlins.
\newblock Searching in the plane.
\newblock {\em Inform. and Comput.}, 106(2):234--252, 1993.

\bibitem{Berman}
Piotr Berman.
\newblock On-line searching and navigation.
\newblock In {\em Online algorithms (Schloss Dagstuhl, 1996)}, volume 1442 of
  {\em Lecture Notes in Comput. Sci.}, pages 232--241. Springer, Berlin, 1998.

\bibitem{BlumRaghavanSchieber}
Avrim Blum, Prabhakar Raghavan, and Baruch Schieber.
\newblock Navigating in unfamiliar geometric terrain.
\newblock {\em SIAM J. Comput.}, 26(1):110--137, February 1997.

\bibitem{Canny}
John Canny.
\newblock {\em The complexity of robot motion planning}.
\newblock The ACM Distinguished Dissertation Series. The MIT Press, 1988.

\bibitem{Caraballo}
David~G. Caraballo.
\newblock Areas of level sets of distance functions induced by asymmetric
  norms.
\newblock {\em Pacific J. Math.}, 218(1):37--52, 2005.

\bibitem{ChosetBurdick}
Howie Choset and Joel Burdick.
\newblock Sensor-based exploration: the {H}ierarchical {G}eneralized {V}oronoi
  {G}raph.
\newblock {\em International Journal of Robotics Research}, 19(2):96--125,
  February 2000.

\bibitem{CoxYap}
James Cox and Chee-Keng Yap.
\newblock On-line motion planning: the case of a planar rod.
\newblock {\em Ann. Math. Artif. Intell.}, 3(1):1--20, 1991.

\bibitem{FiatWoeginger}
Amos Fiat and Gerhard~J. Woeginger, editors.
\newblock {\em Online algorithms: the state of the art}, volume 1442 of {\em
  Lecture Notes in Computer Science}.
\newblock Springer-Verlag, Berlin, 1998.
\newblock Papers from the Workshop on the Competitive Analysis of On-line
  Algorithms held in Schloss Dagstuhl, June 1996.

\bibitem{GabrielyRimon}
Yoav Gabriely and Elon Rimon.
\newblock Cbug: A quadratically competitive mobile robot navigation algorithm.
\newblock Preprint, January 2005.

\bibitem{Latombe}
J.~C. Latombe.
\newblock {\em Robot Motion Planning}.
\newblock Kluwer Academic, Boston, MA, 1991.

\bibitem{LaValle}
Steve Lavalle.
\newblock {\em Planning Algorithms}.
\newblock Cambridge University Press, 2006.

\bibitem{LumelskyStepanov}
Vladimir~J. Lumelsky and Alexander~A. Stepanov.
\newblock Path-planning strategies for a point mobile automaton moving amidst
  unknown obstacles of arbitrary shape.
\newblock {\em Algorithmica}, 2(4):403--430, 1987.
\newblock Special issue on robotics.

\bibitem{PapadimitriouYannakakis}
C.~H. Papadimitriou and M.~Yannakakis.
\newblock Shortest paths without a map.
\newblock {\em Theoretical Computer Science}, 84:127--150, 1991.

\bibitem{ShenNagy}
C.~N. Shen and G.~Nagy.
\newblock Autonomous navigation to provide long distance surface traverses for
  {M}ars rover sample return mission.
\newblock {\em Proc. IEEE Symp. on Inelligent Control}, pages 362--367, 1989.

\end{thebibliography}

\end{document}